\pdfoutput=1
\documentclass[pdftex,pra,aps,twocolumn,twoside,nofootinbib,showkeys,showpacs]{revtex4}

\usepackage{float}
\usepackage{dsfont}
\usepackage{indentfirst}
\usepackage{url}
\usepackage[stable]{footmisc}
\usepackage{fancyhdr}
\usepackage{appendix}
\usepackage{mathtools}
\usepackage[all]{xy}
\usepackage[pdftex]{graphicx}
\usepackage{multirow}
\usepackage{graphicx}
\usepackage{youngtab}
\usepackage[T1]{fontenc}
\usepackage{textcomp}
\usepackage{newcent}
\usepackage[sc,noBBpl]{mathpazo}
\usepackage{amsmath,amsfonts,amssymb,amsthm}
\usepackage[mathscr]{eucal}
\usepackage{cleveref}

\usepackage{tabularx}
\makeatletter
\def\hlinewd#1{%
	\noalign{\ifnum0=`}\fi\hrule \@height #1 %
	\futurelet\reserved@a\@xhline}
\makeatother

\newcommand\scalemath[2]{\scalebox{#1}{\mbox{\ensuremath{\displaystyle #2}}}}
\theoremstyle{plain}
\newtheorem{theorem}{Theorem}
\newtheorem{lemma}[theorem]{Lemma}

\newtheorem{notation}[theorem]{Notation}

\newtheoremstyle{note}{\topsep}{\topsep}{\slshape}{}{\scshape}{}{ }{}
\theoremstyle{note}

\newtheorem{remark}[theorem]{Remark}
\newtheorem{example}[theorem]{Example}
\newcommand\tr{\operatorname{Tr}}

\renewcommand{\>}{\rangle}

\newcommand\be{\begin{equation}}
\newcommand\ee{\end{equation}}
\newcommand\bea{\begin{array}}
	\newcommand\eea{\end{array}}
\newcommand\ben{\begin{eqnarray}}
\newcommand\een{\end{eqnarray}}
\newcommand\ot{\otimes}

\newcommand\bei{\begin{itemize}}
	\newcommand\eei{\end{itemize}}
\newcommand\bee{\begin{enumerate}}
	\newcommand\eee{\end{enumerate}}
\newcommand{\spec}{\operatorname{spec}}

{\par\addvspace{\medskipamount}\noindent\textbf{Examples.}\hspace{1ex}}%
{\par\medskip}

\begin{document}
	\title{A Class of PPT Entangled States Arbitrary Far From  Separable States}
	
	\author{Adam Rutkowski, Micha{\l} Studzi\'nski}
	\affiliation{Institute of Theoretical Physics and Astrophysics, University of Gda\'nsk,National Quantum Information Centre, 80-952 Gda\'nsk, Poland
	}
	
	\date{\today}
	
	\begin{abstract}
		In this paper we show an explicit construction of multipartite class of entangled states with the PPT (Positive Partial Transposition) property in every cut.  We investigate properties of this class of states focusing on the trace distance from the set of separable states.  We provide an explicit sub-class of the  multipartite entangled PPT states which are arbitrary far from the set of separable states.  We argue, that in the multipartite case the mentioned distance increases with dimension of the local Hilbert space. In our construction is  we do not have to use many copies of initial state living on the smaller space to boost the trace distance as in the previous attempts to this problem.
	\end{abstract}
	
	
	\maketitle
	\let\oldthefootnote\thefootnote
	\renewcommand{\thefootnote}{\fnsymbol{footnote}}
	\footnotetext[1]{email: \url{studzinski.m.g@gmail.com}}
	\let\thefootnote\oldthefootnote

The problem of direct construction of PPT entangled states, which are arbitrary far from the set of separable states when increasing the local dimension, was started in the paper~\cite{Badziag_pbits}. Authors described construction which is based on the notion of the private states, which gives the link between quantum entanglement and quantum security (see~\cite{KH_phd} and references within in it).
Using directly properties of PPT private states and tensoring method - taking many copies of a given state - authors constructed class of the PPT states which are  $2-\epsilon$ far from the set of separable states,  whenever the dimension is   $2^{\operatorname{poly}\left(\operatorname{log}(1/\epsilon) \right) }$.
Going further, in the paper~\cite{Rut} authors generalised mentioned result by showing, that to obtain PPT entangled states, arbitrarily  far from the set of separable states,  there is no need to use many copies of the initial state. Namely it is enough to take appropriate mixture of the private states defined on the orthogonal subspaces. Such construction implies, that the distance from the set of separable states scales with the dimension as $d \propto 1/\epsilon^3$ for every $\epsilon >0$, and it improves scaling obtained in~\cite{Badziag_pbits}. Moreover it was shown, that this new class of  states reduces, for specific choices of parameters to many previously known examples of private states.
Here we generalise construction presented in~\cite{Rut} into multipartite case. We show, that there exist states which are maximally far from the set of separable states and PPT entangled. Direct construction of such kind of density matrices is also presented. Up to our best knowledge this is the first result of such kind with further possible applications to the analysis of multipartite linear positive maps~\cite{Glaucia,Huber}.

\section{Notation and Definitions}
\label{two}
Before the formulation of the  main problem  we have to say here a few words more about notation used in this manuscript. In this section and in our further considerations  by   $\mathcal{B}(\mathcal{H})$ we denote the algebra of all bounded linear operators on the Hilbert space $\mathcal{H}$, and we take $\mathcal{H}\cong \mathbb{C}^d$. Using this notation let us define the following set:
\be
\label{Qset}
\mathcal{S}(\mathcal{H})=\{\rho \in \mathcal{B}(\mathcal{H}) \ | \ \rho \geq 0, \  \tr \rho=1\},
\ee
which is the set of all states on a space $\mathcal{H}$. Let us now suppose, that we are given with the state $\rho_{A_1\cdots A_m}$ of $m$ subsystems $A_1,\ldots,A_m$ defined on the Hilbert space $\mathcal{H}_{A_1\cdots A_m}=\mathcal{H}_{A_1}\ot\cdots \ot \mathcal{H}_{A_m}$. Defining the partition $p\equiv \left\lbrace X_1,\ldots,X_k \right\rbrace $, where $X_i$ are disjoint subsets of the indices $X=\{1,\ldots,m\}$, $ \ \cup_{r=1}^k X_r=X$, we can say, that the state $\rho_{A_1\cdots A_m}$ is separable with respect to the partition $p$, if and only if it can be written as
\be
\label{sep1}
\rho_{A_1\cdots A_m}=\sum_{i=1}^Mp_i\rho_{1}^i\ot \cdots \ot \rho_k^i,\qquad \forall i \ p_i \geq 0,\qquad \sum_i p_i=1.
\ee
It is easy to see from the above definition that if $k=n$ we obtain definition of fully separable state, and for $k=m=2$,~\cref{sep1} reduces to the  standard bipartite definition of separability. In this paper we consider the multiparitite density matrices on $\mathcal{B}\left(\mathcal{H}_A^{\ot n}\ot \mathcal{H}_B^{\ot n} \right)$ with $\operatorname{dim}\mathcal{H}_A=d_A$  and $\operatorname{dim}\mathcal{H}_B=d_B$.

In the previous papers treating on the similar problem authors deal with  separable states acting on $\mathcal{H}_{A_1A_2B_1B_2}$ in the cut $A_1A_2:B_1B_2$. In our multipartite approach ($n \geq 2$)  the separable states will be denoted as $\mathcal{SEP}^{(\overline{A},\overline{B})}\equiv \mathcal{\overline{SEP}}$, where $\overline{A}=A_1A_2\ldots A_n$  and similarly for $\overline{B}$.
Moreover, in the further part of this manuscript, whenever we talk about the distance from the set of the separable states we mean the following:
\begin{notation}
	Suppose that we are given with a quantum state $\rho \in \mathcal{B}\left(\mathcal{H}_{\bar{A}}\ot \mathcal{H}_{\bar{B}} \right) $ and the set of separable states $\mathcal{\overline{SEP}}$, then by $\operatorname{dist}\left( \rho,\mathcal{\overline{SEP}}\right)  $ we understand the following quantity
	\be
	\label{dd}
	\operatorname{dist}\left( \rho,\mathcal{\overline{SEP}}\right) =\frac{1}{2}\mathop{\min}\limits_{\sigma \in \mathcal{\overline{SEP}}}||\rho-\sigma||_1,
	\ee
	which is the minimal trace distance. In~\cref{dd} by $||\cdot||_1$ we understand the trace norm which is defined as  $||X||_1\equiv \tr\left[XX^{\dagger} \right]^{1/2}$  for all matrices $X$.
\end{notation}
As we see later, restriction to bi-separability is not a limiting factor in our problem. 

\section{Construction Method}
\label{method}
The main goal of this section is the construction a class of the multipartite entangled operators
\be
\label{1}
\rho \in \mathcal{B}\left(\mathcal{H}_A^{\ot n}\ot \mathcal{H}_B^{\ot n} \right) ,
\ee
having  PPT property with respect to any cut of subsystems and are arbitrary far from the set of fully separable states.

Let  us define $\forall_k  \ \alpha_k \in \{0,1\}$ the binary vector $\mathbold{\alpha}\equiv \left(\alpha_i,\ldots,\alpha_{n-1} \right)$, such that $\mathbold{\alpha} \neq \bf{0}$.
Secondly we define the operators of the partial transposition with respect to an arbitrary  cut as
\be
\tau_{\mathbold{\alpha}}\equiv \text{\noindent
	\(\mathds{1}\)}\ot \operatorname{T}^{\alpha_1} \ot \cdots \ot \operatorname{T}^{\alpha_{n-1}},
\ee
where $\operatorname{T}^{\alpha_k}$ for $1\leq k \leq n-1$ denotes standard operation of transposition on $k^{\text{th}}$ subsystem, whenever $\alpha_k\neq 0$. If for some $k$ there is $\alpha_k=0$ we define $T^{\alpha_k}=\text{\noindent
	\(\mathds{1}\)}_k$.
It is easy to see, that in total we have 
\be
\label{N}
N=\sum_{k=1}^{n-1}\binom{n-1}{k}=2^{n-1}-1
\ee
such operators. Next we introduce the operators
\be
E_{ij}^{\ot n}\equiv \underbrace{e_{ij}\ot \cdots \ot e_{ij}}_{n},
\ee
where $\left\lbrace e_{ij}\right\rbrace_{i,j=1}^d$, with $e_{ij}=|i\>\<j|$ denotes the basis of $n \times n$ complex matrices $M(n,\mathbb{C})$.  Let us consider the following  mixture:
\be
\label{sum1}
\rho=\sum_{l=0}^{ND}\rho^{(l)} \in \mathcal{B}\left(\mathcal{H}_A^{\ot n}\ot \mathcal{H}_B^{\ot n} \right),
\ee
where $D=\frac{1}{2}d_A(d_A-1)$. In the sum given by~\cref{sum1} we can distinguish state $\rho^{(0)}$ in the maximally entangled form
\be
\label{max_form}
\rho^{(0)}=\sum_{i,j=1}^{d_A} E_{ij}^{\ot n}\ot a_{ij}^{(0)},\qquad a_{ij}^{(0)}\in \mathcal{B}\left(\mathcal{H}_B^{\ot n} \right).
\ee
Herein we define the rest of $dN$  operators in such a way, that whole operator $\rho$ admits the PPT property. Observe that $\forall l\neq 0$ there exists a bijection $f:(\mathbold{\alpha},i,j)\rightarrow l$ which allows us to define off-diagonal term of $l^{\text{th}}-$ operators as
\be
\rho_{12}^{(l)}=\tau_{\alpha}\left(E_{ij}^{\ot n} \right) \ot a_{12}^{(l)}, \quad \text{where}\quad a_{12}^{(l)}\in\mathcal{B}\left(\mathcal{H}^{\ot n}_B \right),
\ee
and for given $l$ we have
\be
\rho^{(l)}=\sum_{i,j=1}^2\rho_{ij}^{(l)}.
\ee
Having the form of $\rho_{12}^{(l)}$, together with all sub-blocks $a_{21}^{(l)}, a_{22}^{(l)}, a_{11}^{(l)}\in \mathcal{B}\left(\mathcal{H}_B^{\ot n} \right) $ we have full information about all $\rho^{(l)}$, so we know the structure of $\rho$ from~\cref{sum1}.
 Example below illustrates how above construction works in  practice.
\begin{example}
	Let us consider the case when $n=3$ and the local dimension of the Hilbert space is $d_A=3$, then the total state $\rho \in \mathcal{B}\left( \mathcal{H}_A^{\ot 3} \ot \mathcal{H}_B^{\ot 3}\right) $ from~\cref{1} can be represented as:	
	\begin{widetext}
\be
	\rho=\left(\scalemath{0.65}{\begin{array}{ccc|ccc|ccc !{\vrule width 1.5pt} ccc|ccc|ccc !{\vrule width 1.5pt} ccc|ccc|ccc}
		a_{11}^{(0)} & \cdot & \cdot & \cdot & \cdot& \cdot & \cdot& \cdot & \cdot& \cdot & \cdot& \cdot & \cdot & a_{12}^{(0)} & \cdot & \cdot& \cdot & \cdot& \cdot & \cdot& \cdot & \cdot& \cdot & \cdot & \cdot & \cdot & a_{13}^{(0)}\\
		\cdot & a_{11}^{(1)} & \cdot & \cdot & \cdot& \cdot & \cdot& \cdot & \cdot& \cdot & \cdot& \cdot & a_{12}^{(1)} & \cdot & \cdot& \cdot & \cdot& \cdot & \cdot& \cdot & \cdot& \cdot & \cdot & \cdot & \cdot & \cdot & \cdot \\
		\cdot & \cdot & a_{11}^{(2)} & \cdot & \cdot& \cdot & \cdot& \cdot & \cdot& \cdot & \cdot& \cdot & \cdot & \cdot & \cdot & \cdot & \cdot & \cdot& \cdot & \cdot& \cdot & \cdot& \cdot & \cdot & a_{12}^{(2)} & \cdot & \cdot\\
		\hline
		\cdot & \cdot & \cdot & a_{11}^{(3)} & \cdot& \cdot & \cdot& \cdot & \cdot& \cdot & a_{12}^{(3)}& \cdot & \cdot & \cdot & \cdot& \cdot & \cdot& \cdot & \cdot& \cdot & \cdot& \cdot & \cdot & \cdot & \cdot & \cdot & \cdot \\
		\cdot & \cdot & \cdot & \cdot & a_{11}^{(4)}& \cdot & \cdot& \cdot & \cdot& a_{12}^{(4)} & \cdot& \cdot & \cdot & \cdot & \cdot & \cdot & \cdot & \cdot& \cdot & \cdot& \cdot & \cdot& \cdot & \cdot & \cdot & \cdot & \cdot\\
			\cdot & \cdot & \cdot & \cdot & \cdot& \cdot & \cdot& \cdot & \cdot& \cdot & \cdot& \cdot & \cdot & \cdot & \cdot & \cdot& \cdot & \cdot& \cdot & \cdot& \cdot & \cdot& \cdot & \cdot & \cdot & \cdot & \cdot\\
			\hline
			\cdot & \cdot & \cdot & \cdot & \cdot& \cdot & a_{11}^{(5)}& \cdot & \cdot& \cdot & \cdot& \cdot & \cdot & \cdot & \cdot & \cdot& \cdot & \cdot& \cdot & \cdot& a_{12}^{(5)} & \cdot& \cdot & \cdot & \cdot & \cdot & \cdot \\
			\cdot & \cdot & \cdot & \cdot & \cdot& \cdot & \cdot& \cdot & \cdot& \cdot & \cdot& \cdot & \cdot & \cdot & \cdot& \cdot & \cdot& \cdot & \cdot& \cdot & \cdot& \cdot & \cdot & \cdot & \cdot & \cdot & \cdot \\
			\cdot & \cdot & \cdot & \cdot & \cdot& \cdot & \cdot& \cdot & a_{11}^{(6)}& \cdot & \cdot& \cdot & \cdot & \cdot & \cdot & \cdot & \cdot & \cdot& a_{12}^{(6)} & \cdot& \cdot & \cdot& \cdot & \cdot & \cdot & \cdot & \cdot\\
				\hlinewd{2pt}
			\cdot & \cdot & \cdot & \cdot & a_{21}^{(4)}& \cdot & \cdot& \cdot & \cdot& a_{22}^{(4)} & \cdot& \cdot & \cdot & \cdot & \cdot & \cdot& \cdot & \cdot& \cdot & \cdot& \cdot & \cdot& \cdot & \cdot & \cdot & \cdot & \cdot \\
			\cdot & \cdot & \cdot & a_{21}^{(3)} & \cdot& \cdot & \cdot& \cdot & \cdot& \cdot & a_{22}^{(3)}& \cdot & \cdot & \cdot & \cdot& \cdot & \cdot& \cdot & \cdot& \cdot & \cdot& \cdot & \cdot & \cdot & \cdot & \cdot & \cdot \\
			\cdot & \cdot & \cdot & \cdot & \cdot& \cdot & \cdot& \cdot & \cdot& \cdot & \cdot& \cdot & \cdot & \cdot & \cdot & \cdot & \cdot & \cdot& \cdot & \cdot& \cdot & \cdot& \cdot & \cdot & \cdot & \cdot & \cdot\\
			\hline
			\cdot & a_{21}^{(1)} & \cdot & \cdot & \cdot& \cdot & \cdot& \cdot & \cdot& \cdot & \cdot& \cdot & a_{22}^{(1)} & \cdot & \cdot & \cdot& \cdot & \cdot& \cdot & \cdot& \cdot & \cdot& \cdot & \cdot & \cdot & \cdot & \cdot \\
			a_{21}^{(0)} & \cdot & \cdot & \cdot & \cdot& \cdot & \cdot& \cdot & \cdot& \cdot & \cdot& \cdot & \cdot & a_{22}^{(0)} & \cdot& \cdot & \cdot& \cdot & \cdot& \cdot & \cdot& \cdot & \cdot & \cdot & \cdot & \cdot & a_{23}^{(0)} \\
			\cdot & \cdot & \cdot & \cdot & \cdot& \cdot & \cdot& \cdot & \cdot& \cdot & \cdot& \cdot & \cdot & \cdot & a_{11}^{(7)} & \cdot & \cdot & \cdot& \cdot & \cdot& \cdot & \cdot& \cdot & \cdot & \cdot & a_{12}^{(7)} & \cdot\\
			\hline
			\cdot & \cdot & \cdot & \cdot & \cdot& \cdot & \cdot& \cdot & \cdot& \cdot & \cdot& \cdot & \cdot & \cdot & \cdot & \cdot& \cdot & \cdot& \cdot & \cdot& \cdot & \cdot& \cdot & \cdot & \cdot & \cdot & \cdot \\
			\cdot & \cdot & \cdot & \cdot & \cdot& \cdot & \cdot& \cdot & \cdot& \cdot & \cdot& \cdot & \cdot & \cdot & \cdot& \cdot & a_{11}^{(8)}& \cdot & \cdot& \cdot & \cdot& \cdot & \cdot & a_{12}^{(8)} & \cdot & \cdot & \cdot \\
			\cdot & \cdot & \cdot & \cdot & \cdot& \cdot & \cdot& \cdot & \cdot& \cdot & \cdot& \cdot & \cdot & \cdot & \cdot & \cdot & \cdot & a_{11}^{(9)}& \cdot & \cdot& \cdot & \cdot& a_{12}^{(9)} & \cdot & \cdot & \cdot & \cdot\\
		\hlinewd{2pt}
			\cdot & \cdot & \cdot & \cdot & \cdot& \cdot & \cdot& \cdot & a_{21}^{(6)}& \cdot & \cdot& \cdot & \cdot & \cdot & \cdot & \cdot& \cdot & \cdot& a_{22}^{(6)} & \cdot& \cdot & \cdot& \cdot & \cdot & \cdot & \cdot & \cdot \\
			\cdot & \cdot & \cdot & \cdot & \cdot& \cdot & \cdot& \cdot & \cdot& \cdot & \cdot& \cdot & \cdot & \cdot & \cdot& \cdot & \cdot& \cdot & \cdot& \cdot & \cdot& \cdot & \cdot & \cdot & \cdot & \cdot & \cdot \\
			\cdot & \cdot & \cdot & \cdot & \cdot& \cdot & a_{21}^{(5)}& \cdot & \cdot& \cdot & \cdot& \cdot & \cdot & \cdot & \cdot & \cdot & \cdot & \cdot& \cdot & \cdot & a_{22}^{(5)} & \cdot& \cdot & \cdot & \cdot & \cdot & \cdot\\
			\hline
			\cdot & \cdot & \cdot & \cdot & \cdot& \cdot & \cdot& \cdot & \cdot& \cdot & \cdot& \cdot & \cdot & \cdot & \cdot & \cdot& \cdot & \cdot& \cdot & \cdot& \cdot & \cdot& \cdot & \cdot & \cdot & \cdot & \cdot \\
			\cdot & \cdot & \cdot & \cdot & \cdot& \cdot & \cdot& \cdot & \cdot& \cdot & \cdot& \cdot & \cdot & \cdot & \cdot& \cdot & \cdot & a_{21}^{(9)} & \cdot& \cdot & \cdot& \cdot  & a_{22}^{(9)} & \cdot & \cdot & \cdot & \cdot \\
			\cdot & \cdot & \cdot & \cdot & \cdot& \cdot & \cdot& \cdot & \cdot& \cdot & \cdot& \cdot & \cdot & \cdot & \cdot & \cdot & a_{21}^{(8)} & \cdot& \cdot & \cdot& \cdot & \cdot& \cdot & a_{22}^{(8)} & \cdot & \cdot & \cdot\\
			\hline
			\cdot & \cdot & a_{21}^{(2)} & \cdot & \cdot& \cdot & \cdot& \cdot & \cdot& \cdot & \cdot& \cdot & \cdot & \cdot & \cdot & \cdot& \cdot & \cdot& \cdot & \cdot& \cdot & \cdot& \cdot & \cdot & a_{22}^{(2)} & \cdot & \cdot \\
			\cdot & \cdot & \cdot & \cdot & \cdot& \cdot & \cdot& \cdot & \cdot& \cdot & \cdot& \cdot & \cdot & \cdot & a_{21}^{(7)}& \cdot & \cdot& \cdot & \cdot& \cdot & \cdot& \cdot & \cdot & \cdot & \cdot & a_{22}^{(7)} & \cdot \\
			a_{31}^{(0)} & \cdot & \cdot & \cdot & \cdot& \cdot & \cdot& \cdot & \cdot& \cdot & \cdot& \cdot & \cdot & a_{32}^{(0)} & \cdot & \cdot & \cdot & \cdot& \cdot & \cdot& \cdot & \cdot& \cdot & \cdot & \cdot & \cdot & a_{22}^{(0)}\\
		\end{array}}\right),
	\ee
	\end{widetext}
	where dots denote zeros and $a_{ij}^{(l)}\in \mathcal{B}\left(\mathcal{H}_B^{\ot 3} \right)$.
\end{example}
\begin{remark}
From all considerations presented in this section, we see that in particular case, when $n=2$, for particular choice of parameters, we reduce our class of operators to already known examples of pbits ($d_A=2$) and pdits ($d_A=3$), where now $d_A$ plays role of dimension of the key part. For the detailed description we refer reader to~\cite{KH_phd}.
\end{remark}

\begin{remark}
Let us notice that in the recent paper~\cite{Huber} authors consider the similar operator as in~\eqref{sum1} (in unnormalised form it is equivalent to our operator up to swap operation with respect to subsystems) to show that some linear positive maps are non-decomposable.
\end{remark}

\section{Properties of given subclass of states}
To provide our main result we need to consider the following mixture of states:
\be
\label{form}
\rho = p\rho^{(0)}+\frac{q}{ND}\sum_{l=1}^{ND}\rho^{(l)},
\ee
where $p+q=1$, $D=\frac{1}{2}d_A(d_A-1)$ and $N$ is given through eq.~\eqref{N}. Here and later we assume operators $\rho^{(0)},\rho^{(l)}$ to be normalised. Let us notice that  states given in~\cref{form} belongs to the class of the states defined in~\eqref{sum1}. Our first results says:

\begin{lemma}
	\label{q}
	Let us assume that we are given with $\rho\in \mathcal{B}\left(\mathcal{H}_A^{\ot n}\ot \mathcal{H}_B^{\ot n} \right) $ as in~\cref{form}, and the state $\rho^{(0)}$  is in entangled form~\eqref{max_form}, then the following statement holds
	\be
	||\rho-\rho^{(0)}||_1=2q.
	\ee	
\end{lemma}

\begin{proof}
	We prove above statement by direct calculations, namely we have
	\be
	\begin{split}
	||\rho-\rho^{(0)}||_1&=\left| \left| p\rho^{(0)}+\frac{q}{ND}\sum_{l=1}^{ND}\rho^{(l)}-\rho^{(0)}\right| \right| _1\\
	&=\frac{q}{ND}\left| \left| \sum_{l=1}^{ND}\rho^{(l)}-ND\rho^{(0)}\right| \right| _1.
	\end{split}
	\ee
	At this point we use definition of the trace norm and rewrite above equation as
	\be
	\label{d1}
	\begin{split}
	&||\rho-\rho^{(0)}||_1=\\
	&=\frac{q}{ND}\tr\left[\left(\sum_{l=1}^{ND}\rho^{(l)}-ND\rho^{(0)} \right) \left(\sum_{l=1}^{ND}\rho^{(l)}-ND\rho^{(0)} \right)^{\dagger}  \right]^{1/2}.
	\end{split}
	\ee
	Because all operators in~\cref{d1} are hermitian and supported on orthogonal subspaces, we reduce~\cref{d1} to
	\be
	||\rho-\rho^{(0)}||_1=\frac{q}{ND}\tr\left[\left(\sum_{l=1}^{ND}\rho^{(l)}+ND\rho^{(0)} \right)^2 \right]^{1/2}.
	\ee
	Making further simplification, finally we obtain:
	\be
	||\rho-\rho^{(0)}||_1=\frac{q}{ND}\tr\left[\sum_{l=1}^{ND}\rho^{(l)}+ND\rho^{(0)} \right]=2q.
	\ee
\end{proof}
Before we go further in our considerations, we make here some specific choice of the interior structure for the states from~\cref{form}. Namely, we assume the following PPT-invariant form of the sub-blocks:
\be
\label{s1}
\begin{split}
&a_{ij}^{(0)}=a=\frac{1}{d_Ad_B^{n}}\begin{pmatrix}
1 & 0 & \cdots & 0 & 0\\
0 & 1 & \cdots & 0 & 0\\
\vdots & \vdots & \ddots & \vdots & \vdots \\
0 & 0 & \cdots & 1 & 0\\
0 & 0 & \cdots & 0 & 1
\end{pmatrix}_{d_B^n \times d_B^n}\\ &a_{rs}^{(l)}=b=\frac{1}{2d_B^{n}}\begin{pmatrix}
1 & 1 & \cdots & 1 & 1\\
1 & 1 & \cdots & 1 & 1\\
\vdots & \vdots & \ddots & \vdots & \vdots \\
1 & 1 & \cdots & 1 & 1\\
1 & 1 & \cdots & 1 & 1
\end{pmatrix}_{d_B^n \times d_B^n},
\end{split}
\ee
where $1\leq i,j \leq d_A$ and  for every $ 1\leq l\leq ND$ we have $1 \leq r,s \leq 2$. Matrices $a,b$ have the following spectra:
\be
\label{spectra}
\begin{split}
&\spec\left(a\right)=\left\lbrace \frac{1}{d_Ad_B^n},\ldots, \frac{1}{d_Ad_B^n} \right\rbrace,\\
&\spec\left(b \right) =\left\lbrace \frac{1}{2},0,\ldots,0 \right\rbrace,
\end{split}
\ee
where $0$ in $\spec (b)$ is taken $d_B^n-1$ times.
Now we are in the position to formulate and prove the following:
\begin{lemma}
	\label{sq}
	Let us consider class of states $\rho$ given by~\cref{form}, together with specific choice of sub-blocks from~\cref{s1}. Then the parameter $q$ describing the trace distance in Lemma~\ref{q} is equal to
	\be
	q=\frac{1}{1+\frac{d_B^n}{N(d_A-1)}},
	\ee
	where $N$ is given by expression~\eqref{N}.
\end{lemma}

\begin{proof}	
	To get the result we need to consider positivity of the state $\rho$ from~\cref{form} before and after partial transposition with respect to every possible cut. Then we use the property that all sub-blocks given in~\cref{s1} are PPT invariant, so their spectra remain unchanged after that partial transposition. Since all the states taken for the construction of $\rho$ are supported on the orthogonal subspaces checking positive semidefiniteness  of $\rho$ is equivalent to checking the  positive semidefiniteness of the following matrix
	\be
	\label{X}
	X=\begin{pmatrix}
	xb & pa\\
	pa & xb
	\end{pmatrix}\geq 0,
	\ee
		where $x=\frac{q}{ND}$, $p=1-q$, $D=d_A(d_A-1)/2$ and real matrices $a,b$ are given in by~\cref{s1}. Using analogous observation we deduce condition for positive semidefiniteness of $\rho$ after partial transposition:
	\be
	\label{Y}
	Y=\begin{pmatrix}
		pa & xb & \cdots & xb \\
		xb & pa & \cdots & xb \\
		\vdots & \vdots & \ddots & \vdots\\
		xb & xb & \cdots & pa
	\end{pmatrix} \geq 0.
	\ee
 Matrices $X,Y$ can be decomposed as
	\be
	\label{decomp}
	\begin{split}
		X&=\text{\noindent
			\(\mathds{1}\)}_2\ot xb + \mathbb{I}_2\ot pa - \text{\noindent
			\(\mathds{1}\)}_2\ot pa,\\
		Y&=\text{\noindent
			\(\mathds{1}\)}_n\ot pa + \mathbb{I}_n\ot xb - \text{\noindent
			\(\mathds{1}\)}_n\ot xb.\\
	\end{split}
	\ee
	In~\cref{decomp} by $\text{\noindent
		\(\mathds{1}\)}_n, \text{\noindent
		\(\mathds{1}\)}_2$ we denote $n \times n$, $2 \times 2$ identity matrices respectively By $\mathbb{I}_n,\mathbb{I}_2$ we denote $n \times n, 2\times 2$ matrices filled with ones only. Matrix $\mathbb{I}_n$ has only one non-zero eigenvalue equal to $n$. Because all terms in the decomposition given by~\cref{decomp} commute we can write positivity conditions from~\eqref{X},~\eqref{Y} using components of~\eqref{decomp}:
	\be
	\label{or}
	\begin{split}
	\lambda(X)&=x\lambda(b)+2p\lambda(a)-p\lambda(a)\geq 0,\\ \lambda(X)&=x\lambda(b)-p\lambda(a)\geq 0,\\
	\lambda(Y)&=p\lambda(a)+nx\lambda(b)-x\lambda(b)\geq 0,\\ \lambda(Y)&=p\lambda(a)-x\lambda(b)\geq 0.
\end{split}
	\ee
	First and the third inequality in~\eqref{or} are always satisfied. Non-trivial condition we get from the second and the fourth inequality getting
	\be
	p\lambda(a)-\frac{q}{N(d_A-1)}\lambda(b)= 0.
	\ee
	Using constraint $p+q=1$ and $D=\frac{1}{2}d_A(d_A-1)$ we derive the parameter $q$:
	\be
	q=\frac{1}{1+\frac{1}{ND}\frac{\lambda(b)}{\lambda(a)}}.
	\ee
	Inserting explicit form of eigenvalues $\lambda(a),\lambda(b)$ given in~\eqref{spectra} we get:
	\be
	q=\frac{1}{1+\frac{d_B^n}{N(d_A-1)}},
	\ee
	where the number $N$ is given in expression~\eqref{N}.
	This finishes the proof.
\end{proof}
The next goal of this section is to compute the trace distance between set of the separable states $\mathcal{\overline{SEP}}$ defined as in~\Cref{method}. Namely, we have the following
\begin{lemma}
	\label{fromsep}
The trace distance between class of multipartite states given by
\be
\label{bla}
\rho=p\rho^{(0)}+\frac{q}{ND}\sum_{l=1}^{ND}\rho^{(l)},
\ee
with sub-blocks from~\eqref{s1}  is bounded from below as
\be
\operatorname{dist}\left(\rho,\mathcal{\overline{SEP}} \right)\geq 1-\frac{1}{d_A^n}-\frac{1}{1+\frac{d_B^n}{N(d_A-1)}}.
\ee
 In the above we take $D=\frac{1}{2}d_A(d_A-1)$, $N=2^{n-1}-1$.
\end{lemma}

\begin{proof}
Adopting the result from~\cite{Badziag_pbits} we arrive at:
\be
\label{sepB}
\operatorname{dist}\left(\rho^{(0)}, \mathcal{\overline{SEP}}\right) \geq 1-\frac{1}{d_A^n}.
\ee
Now let us take the closest separable state $\sigma_{sep}$ to the state $\rho$ given in~\cref{form} with~\cref{s1}, then we have
\be
\label{triangle}
\begin{split}
&||\rho-\sigma_{sep}||_1+||\rho-\rho^{(0)}||_1 \geq \\ &||\sigma_{sep}-\rho^{(0)}||_1\geq 1-\frac{1}{d_A^n},
\end{split}
\ee
so
\be
||\rho-\sigma_{sep}||_1 \geq 1-\frac{1}{d_A^n}-||\rho-\rho^{(0)}||_1.
\ee
In the next step we can use the statement of~\Cref{sq} taking the worst possible choice of $q=1/(1+d_B^n/N(d_A-1))$
\be
\label{final}
\begin{split}
||\rho-\sigma_{sep}||_1\geq 1-\frac{1}{d_A^n}-\frac{1}{1+\frac{d_B^n}{N(d_A-1)}}.
\end{split}
\ee
This finishes the proof.
\end{proof}

\begin{remark}
One can see that since $\overline{\mathcal{SEP}}\subset \mathcal{SEP}$, expression~\eqref{bla} from Lemma~\ref{fromsep} gives us lower bound on the trace distance also from fully separable states as well as from any set of partially separable states.
\end{remark}

Let us notice, that for the generic case for qubits states i.e. $n=2, d_A=2, d_B=2$ our bound achieves minimum.
At the end of this section lest us consider the scenario in which the multipartite PPT sate $\rho$ acts on $\mathcal{B}(\mathbb{C}^d \ot \mathbb{C}^d)$, where $d=d_A^nd_B^n$. Now we show that states from our class can be arbitrary far from the set of separable states $\mathcal{\overline{SEP}}$ in the mentioned cut for the fixed number of subsystems $n$.  Namely we have the following:
\begin{theorem}
	\label{maxdist}
	For every $\epsilon >0 $ there exists PPT state $\rho$ given by~\cref{form} with~\cref{s1} acting on $\mathcal{B}(\mathbb{C}^d \ot \mathbb{C}^d)$, such that
	\be
	\label{con1}
	\operatorname{dist}\left(\rho, \mathcal{\overline{SEP}}\right) \geq 1-\epsilon,
	\ee
	for $d\leq C(n)/\epsilon^{2+1/n}$, where $C(n)=8N\sqrt[n]{2}$.
\end{theorem}
\begin{proof}
In order to show existence of PPT state which satisfies condition from~\cref{con1} it is enough to consider the worst case scenario by taking $\epsilon/2=1/d_A^n$, and $\epsilon/2=1/(1+d_B^n/(N(d_A-1)))$. Then we get
\be
\begin{split}
d_A^n&=\frac{2}{\epsilon},\\ d_B^n&=N\frac{2-\epsilon}{\epsilon}\left(\sqrt[n]{\frac{2}{\epsilon}}-1 \right).
\end{split}
\ee
Now computing $d=d_k^nd_s^n$ we have
\be
d=d_k^nd_s^n=4N\frac{2-\epsilon}{\epsilon^2}\left(\sqrt[n]{\frac{2}{\epsilon}}-1 \right)\leq \frac{C(n)}{\epsilon^{2+1/n}},
\ee
where $C(n)=8N\sqrt[n]{2}$.
\end{proof}

\begin{remark}
	\label{rem10}
To achieve an arbitrary large distance we do not have to consider so called tensoring (or boosting - taking many copies of state) introduced in~\cite{Badziag_pbits}, it is enough only to increase dimension  $d_B$ for fixed dimension  $d_A$. This is generalizing of theorem given in the paper~\cite{Rut}.
\end{remark}

\begin{remark}
	Theorem~\ref{maxdist} can be easily generalized to different cut, i.e. states acting on $\mathcal{B}(\mathbb{C}^d \ot \mathbb{C}^{d'})$, where $d'\neq d$.
\end{remark}

\section{Summary}
In this paper we present a wide class of multipartite entangled density operators which are PPT invariant with respect to every cut, arbitrary far from the set of separable states. Our method of construction in based on mixing multipartite states defined on orthogonal spaces and imposing set of constraints on positive semidefiniteness of the final mixed state before and after the partial transpositions.  We discuss distance of any state from our class from the set of separable states as a function of mixing parameter $q$. Afterwards we show, that we are able to construct special, but non-trivial subclass of states, for which we evaluate mentioned distance as the function of dimension $d_B$. As we stated in Remark~\ref{rem10} the novelty of our approach lies in the fact, that to boost the distance between states from our class and the set of separable states we do not have to use many copies of them like in the previous approaches.
 Namely we prove, that mentioned distance decreases with growing dimension $d_B$ for fixed $n$. Speaking more precisely, we show that for every $\epsilon >0$, we can find entangled PPT state from our class which is $1-\epsilon$ far from the set of separable states. The scaling of $\epsilon$ with the dimension is  $d\propto 1/\epsilon^{2+1/n}$.
\section*{Acknowledgments}
 MS is supported by grant "Mobilno{\'s}{\'c} Plus IV", 1271/MOB/IV/2015/0 from Polish Ministry of Science and Higher Education. 
\bibliographystyle{unsrt}
\bibliography{References6}
\end{document}